\title{Online embedding of metrics\thanks{This is an expanded and corrected
version of the paper by the same name that appeared in SWAT'20.}}
\author{Ilan Newman\thanks{Department of Computer Science, University
of
Haifa, Haifa, Israel. E-mail: {\tt ilan@cs.haifa.ac.il}. 
This Research
was supported by The Israel Science Foundation, grant number
497/17}, $~ $
$ ~ $Yuri Rabinovich\thanks{Department of Computer Science, University
 of  Haifa, Haifa, Israel. E-mail: {\tt yuri@cs.haifa.ac.il}}}

\date{\today}
\documentclass[11pt]{article}
\usepackage{fullpage}
\usepackage{xspace}
\usepackage{latexsym}
\usepackage{times}
\usepackage{amsfonts}
\usepackage{graphicx}
\usepackage{titling}
\usepackage{imakeidx}
\DeclareMathAlphabet{\mathbbold}{U}{bbold}{m}{n}
\usepackage{amsmath, amsthm, amssymb}
\usepackage{verbatim}
\usepackage[table]{xcolor}
\usepackage{enumerate}
\usepackage{color,graphics}
\usepackage{comment} 
\usepackage{outlines}
\usepackage{mathtools}
\usepackage{microtype}
\usepackage{pgfplots}   
\pgfplotsset{compat=1.7}
\usepackage{pgf,tikz,tikz-qtree}   %
\usetikzlibrary{arrows, automata}  
\usetikzlibrary {positioning}
\usetikzlibrary{automata,positioning,arrows,matrix,backgrounds,calc}
\usetikzlibrary{decorations.text}
\usetikzlibrary{decorations.pathmorphing}
\usetikzlibrary{shapes.geometric}
\usepackage[utf8]{inputenc}
\usepackage{algorithmic}
\usepackage{algorithm}
\usepackage{multirow}
\usepackage{stackrel}
\usepackage{ifthen}
\usepackage{xparse}

\definecolor{Blue}{rgb}{0.3,0.3,0.9}
\definecolor{Black}{rgb}{0.0,0.0,0.0}

\newcommand{\ignore}[1]{}
\def \qed {\hspace*{0pt} \hfill {\quad \vrule height 1ex width 1ex depth 0pt}
 \medskip}

\newcommand{\R}{\ensuremath{\mathbb R}}
\newcommand{\Z}{\ensuremath{\mathbb Z}}
\newcommand{\N}{\ensuremath{\mathbb N}}

\newtheorem{theorem}{Theorem}[section]
\newtheorem{definition}[theorem]{Definition} 
\newtheorem{claim}[theorem]{Claim}

\newtheorem{lemma}[theorem]{Lemma}

\newtheorem{corollary}[theorem]{Corollary}

\def\dist {{\rm dist}}
\def\father{{\rm father}}

\def \S {{ \mathbb S }}

\def \noproof
 {\hspace*{0pt} \hfill {\quad \vrule height 1ex width 1ex depth 0pt}}

\def\wp {{\widetilde{\psi}}} 

\def\dist {{\rm dist}}

\def\ord{{\rm t}}
\def\into{\hookrightarrow}

\newcommand{\D}[1]{D_{#1}}

\def\into{\hookrightarrow}
\begin{document}
\maketitle

\begin{abstract}
We study deterministic online embeddings of metrics spaces into normed spaces
and into trees against an adaptive adversary. Main results include a polynomial 
lower bound on the (multiplicative) distortion of embedding into Euclidean spaces, 
a tight exponential upper bound on embedding into the line, and a 
$(1+\epsilon)$-distortion embedding in $\ell_\infty$ of a suitably high dimension.
\end{abstract}
%
\section{Introduction}
The modern theory of low-distortion embeddings of finite metrics
spaces into various host spaces began to take shape with the appearance of classical results
of Johnson and Lindenstrauss~\cite{JL}\footnote{\em Any $n$-point
  Euclidean metric can be efficiently embedded into
  $\ell_2^{{{\log n} / {\epsilon^2}}}$ with
  $(1+\epsilon)$-distortion.}  and Bourgain~\cite{Bou}\footnote{ \em Any
  $n$-point metric can be efficiently embedded into Euclidean
  space with distortion $O(\log n)$.}, in the last
decades of the 20'th century. It was soon realized that this
theory provides powerful tools for numerous theoretical and practical
algorithmic problems. Nowadays, it is a well developed area deeply related 
to various problems of algorithmic design.

In this paper we study a relatively neglected aspect of metric
embeddings, the {\em online} embeddings.
In this setting, the vertices of the input finite metric space $(X,d)$ are exposed one by one, together with their distances 
to the previously exposed vertices. Each vertex $v$ is
mapped to the host space $({\cal H}, d_{\cal H})$ upon its exposure without altering the
mappings of the previously exposed vertices.
The quality of the resulting embedding $\phi: X \rightarrow \cal H$ is measured by its  {\em expansion} and {\em contraction}:
\[
{\rm expansion:}~~~~\max_{v, u \in X} {{d_{\cal H} (\phi(v),\phi(u))} \over {d(v,u)}}~~~~~~~~~~~~~
{\rm contraction:}~~~~\max_{v, u \in X} {{d(v,u)} \over {d_{\cal H} (\phi(v),\phi(u))}}
\]
The product of the two is called the (multiplicative) {\em
  distortion} of $\phi$. The distortion $\dist(d \into d_{\cal H})$ of
embedding $(X,d)$ into $({\cal H}, d_{\cal H})$ is the minimum
possible distortion over all such mappings $\phi$. Since usually (and in this
paper in particular) the host space is scalable,
$\dist(d \into d_{\cal H})$ can be alternatively defined in the
offline setting as the minimum possible contraction
over non-expanding mappings, or he minimum possible expansion over 
non-contracting mappings.

Since in the online setting the scaling cannot be applied in retrospect,
it is quite possible that the algorithms restricted to producing non-contracting,
non-expanding, or unrestricted embeddings, may all perform differently as measured
by the distortion they incur.

In this paper we focus only on {\it deterministic embeddings} against an {\em adaptive adversary} into the standard
normed spaces $\ell_2, \ell_1, \ell_\infty$ of various dimensions, and
into the trees. Our results clarify what can be achieved in dimension one, and
in dimension exponential in $n$. What happens in between is a
challenging open problem. We also present  a lower bound on online
embedding of a size-$n$ metrics into $\ell_2$ of unbounded dimension.
It is our hope that the findings of the present paper may provide a good starting point for further studies of online embeddings. 

In addition to deterministic online embedding algorithms, it is natural to consider  
significantly more relaxed notion of {\em probabilistic} online embedding against {\em non-adaptive adversary}. In this setting, instead of considering the distortion between $d$ and $d_{\cal H}$, one considers distortion between $d$ and 
$E[d_{\cal H}]$, the expectation over the host metrics produced by a random embedding. Unlike before, the adversary cannot see the partial embeddings, and thus must prepare the hard metric in advance. In this relaxed setting much stronger results were previously obtained, as discussed in the next section.
\subsection{Previous Work}
To the best of our knowledge, the first result about online embedding
has appeared implicitly in the paper ~\cite{NR} of the present authors. The
main result there is a  $\Omega(\sqrt{n})$ lower bound on the distortion 
of an {\em offline} embedding of a shortest path metric of a certain family of serious-parallel graphs $\{D_n\}$ of size exponential in $n$ into $\ell_2$. Without ever mentioning the term "online", the proof,  in fact, establishes a lower bound of $\sqrt{n}$ on the distortion of an online embedding of the shortest path metric of a certain simple family of subgraphs ${G_{2n}}$ of $D_n$ with $2n$ vertices.
Although the paper~\cite{NR} did receive due attention, and its online implications were understood by the authors of~\cite{IMSZ}\footnote{It served as partial motivation for their paper. [Private communication from the authors]},     
the explicit statement, which was and still remains the state of the art
in its context,  has never appeared in print, and went largely unnoticed. Here we amend this situation. This is done in Section~\ref{sec:l2}. \\ 
$\mbox{}$

Another related result has appeared in~\cite{ABLNRRV}\; (Th.~3.1)\; in a rather unrelated context.  
Let $(X,d)$ be an arbitrary metric space with $|X|=n$. Assume that $X$ is exposed in a random uniform order. Then the greedy online algorithm that attaches each new point $v$ to the closest one among the points exposed so far, say $u$, by an edge of length $d(v,u)$, produces
a random dominating tree $T$ so that ${\rm E}[d_T]$ expands $d$ by $O(n^2)$.     

If the order is fixed, a similar but simpler analysis implies that
$d_T$ expands $d$ by at most $2^{n}$. 
Since this turns out to be rather tight 
(up to the basis of the exponent) for a deterministic embedding into a
tree, we shall discuss it in more details in Section~\ref{sec:tree}. \\
$\mbox{}$ 

The first published paper explicitly dedicated to online embeddings is~\cite{IMSZ}.  
Observing that a large part of the offline embedding procedure from~\cite{Barl} can be implemented 
online, ~\cite{IMSZ} has established quite strong results for
probabilistic online embeddings. Most of these results depend on the
so called {\em aspect ratio} $\Delta$ of the input metric $d$, that is, the ratio between the largest and the smallest distance in it. The main results of~\cite{IMSZ} are as follows (it is assumed that $|X|=n$):
\begin{enumerate}
\item A metric space $(X,d)$ can be probabilistically online embedded into $\ell_p^{\log n \cdot (\log \Delta)^{1/p}}$~ with distortion $O(\log n \cdot \log \Delta)$ for any $p \in [1,\infty]$. 

For $p=\infty$, $(X,d)$ can also be embedded
in $\ell_\infty^{\log^{O(1)} n}$ with distortion $O(\log n \cdot \sqrt{\log \Delta})$.  

\item A metric space $(X,d)$ can be probabilistically online embedded into a distribution of a non-contracting ultrametics (and subsequently tree-metrics) with distortion $O(\log n \cdot \log \Delta)$.

On the negative side, $(X,d)$ cannot be probabilistically online embedded into a distribution of a non-contracting 
ultrametics with distortion better than ~$\min \{ n, \log \Delta \}$.
\end{enumerate} 
The subsequent study~\cite{BarFan} entends the results of ~\cite{IMSZ},
and, in particular, shows a lower bound of $\widetilde{\Omega}(\log n \cdot \log \Delta)$ for distortion of  probabilistic embeddings into trees.
\subsection{Our Results}
We are interested in online embeddings into normed spaces, and in particular in the interplay between the distortion and the dimension of the host space. Unlike~\cite{IMSZ,BarFan}, we seek bounds independent of the 
aspect ratio, consider only deterministic embeddings, and assume an adaptive adversary.
\begin{enumerate}
\item {\bf Into $\ell_2$\,:}  There exists a family of metrics $\{d_{2n}\}$ 
such that each $d_{2n}$ (a metric on $2n$ points) requires distortion
$\sqrt{n}$
in any deterministic online embedding into $\ell_2$ of {\em any dimension}. 
The metrics $\{d_{2n}\}$ are the shortest-path metrics of a family $\{G_{2n}\}$ of weighted series-parallel graphs. These metrics are quite simple; e.g., they embed into the line with a constant universally bounded distortion. 

By John's Theorem from the theory of finite-dimensional normed spaces, this implies an $\sqrt{n/D}$ lower bound on
online embedding of $d_{2n}$ into any normed space of dimension $D$. 
\item {\bf Into trees, and into the line.} As mentioned above, a simple greedy online embedding algorithms results in a dominating tree whose metric 
distorts the input $d$ by at most $2^{n}$. We show that this is tight up to the basis of the exponent, and establish a lower bound of $\Omega(2^{n \over 2})$. The "hard" metrics used in the proof are in fact submetrics of a (continuous) cycle, and they embed (offline) into the line with a constant universally bounded distortion.  

Embedding into the line is considerably harder. We provide an online algorithm achieving a distortion of $O(n\cdot 4^n)$. 
\item{\bf Distortion and dimension.} What is the smallest dimension $D$ such that $d$ can be embedded into $\ell_\infty^D$ with distortion at most $(1+\epsilon)$?
Somewhat surprisingly, it turns out that even 4-points metrics (that always isometrically embed offline into $\ell_\infty^3$) require dimension  $D = \Omega({1\over {\epsilon}})$.  

On the positive side, we (efficiently) prove that $D=\left({{cn} \over {\epsilon}}\right)^n$, where $c$ is a suitable universal constant, suffices 
for $\epsilon \leq 1$. While such result would seem to be quite expected, its proof requires some technical effort.
\item {\bf Isometric online embeddings.} We show that size-$n$ tree metrics  (i.e., arbitrary submetrics of the shortest-path metrics of weighted trees) can be isometrically online embed into $\ell_1^{n-1}$. This implies, e.g., 
that such metrics can be online embedded in $\ell_2^{n-1}$ with distortion
$\sqrt{n}$, and isometrically online embedded into $\ell_\infty^{2^{n-2}}$ for $n>1$. 
\end{enumerate} 
{\bf Open questions:}
The results of this paper seem to suggest that in contrast to the classical offline setting, where "good" embeddings have distortion polylogarithmic in $n$, in the online setting "good" embedding have distortion polynomial in $n$. We conjecture that a polynomial distortion is indeed achievable, 
and moreover, a polynomial dimension is always sufficient for the task. It cannot be ruled out at the moment that even a constant dimension will do.
\section{A lower bound for embeddings into $\ell_2$}
\label{sec:l2}
As mentioned in the previous section, the following theorem is implied by the proof of the main result 
of~\cite{NR}. For the sake of simplicity, we restrict the discussion to non-contracting embeddings. A removal of this restriction requires but a slight modification of the argument. 
\begin{theorem} 
\label{th:NR}
There is a family of metrics $\{d_{2n}\}$ on $2n$ points for any natural $n$, 
that requires expansion $\geq \sqrt{n}$ in any non-contracting online embedding into $\ell_2$ of
any (infinite included) dimension. 
\end{theorem}
\begin{proof} {\bf (Sketch)}~~ Given an online non-contracting embedding algorithm ${\cal A}$, the "hard" $\{d_{2n}\}$ is constructed as follows.
It will be the shortest-path metric of the following weighted graph $G_{2n}$. 
$G_2$ is simply the unit-weighted $K_2$. The graph $G_{2n}$ is obtained from $G_{2n-2}$ by choosing
an edge $e=(v,u)$ of weight $2^{2-n}$ in $G_{2n-2}$, and replacing it by a $4$-cycle $\bf v$-$x$-$\bf u$-$y$-$\bf v$
with edges of weight $2^{1-n}$. I.e., a pair of new vertices $x,y$ is exposed, together with their distances
to the previously exposed vertices. 

It remains to specify the edge $e$. The adversary will choose an edge $e=(v,u)$ of weight $2^{2-n}$ in $G_{2n-2}$ expanded by ${\cal A}$ by at least a factor of  $\sqrt{n-1}$. The existence of such 
an edge is proven by induction. It clearly exists for $n=2$. Let $\phi$ be the embedding produced by ${\cal A}$. By the parallelogram law for Euclidean spaces,
\[
\|\phi(x)-\phi(v)\|_2^2 + \| \phi(x)- \phi(u) \|_2^2 + \|\phi(y) - \phi(u)\|_2^2  + \|\phi(y)-\phi(v)\|_2^2 ~~\geq~~ \|\phi(v)- \phi(u) \|_2^2 + \|\phi(x)- \phi(y) \|_2^2\,.
\]
Since by inductive assumption ~ $\|\phi(v)- \phi(u) \|_2 \geq
\sqrt{n-1}\cdot 2^{2-n}$, and since $\phi$ is non-contracting, ~ $\|\phi(x)- \phi(y) \|_2 \geq 2^{2-n}$, one concludes that one of the four new edges
$(v,x)$, $(x,u)$, $(u,y)$ and $(y,u)$, must be expanded by ${\cal A}$ by at least $\sqrt{n}$.

As mentioned above, the metrics $d_{2n}$ are very simple. E.g., it is an easy matter to verify that each 
$d_{2n}$ (offline) embeds into the line with distortion $\leq 3$, and isometrically embeds into $\ell_1$.
\end{proof}
Currently, we do not know how tight is the above bound, and whether is it at all possible
to obtain a polynomially in $n$ small distortion for an online embedding into $\ell_2$.
We do know that this is possible for tree metrics, in view of Theorem~\ref{lem:l-5} below.

To conclude this section, let us comment that the argument used here can be employed for any uniformly
convex space, and not just $\ell_2$. See, e.g., the generalization of~\cite{NR} by~\cite{LN}. 
%
%
%
%
%
\section{Online embedding into trees}\label{sec:tree}
By embedding into trees we mean an online embedding that constructs a tree 
whose vertices may contain {\it Steiner points}. That is, the constructed tree,
besides the points corresponding to the input metric,
may contain additional points. At each step, once a new vertex  is
exposed, the embedding algorithm either chooses an existing vertex, or creates a Steiner point, and 
attaches to it the new vertex by a new edge of a suitable weight. 

\begin{theorem}
\label{cl:tree-emb}
Any metric space on $n$ points can be deterministically online embedded into a
tree with distortion $\leq 2^{n-1}-1$, even without using  Steiner points. 
A priori knowledge of $n$ is not required.
\end{theorem}
\begin{proof}
The algorithm greedily connect the new point $v$ to the previously exposed point $u$ 
that is closest to $v$ in the metric $d$ by an edge of weight $d(v,u)$
 
The analysis is essentially the same as in~\cite{ABLNRRV}. Let $\widetilde{d}$
denote the tree metric approximating $d$. Clearly, $\widetilde{d}$ is not-contracting.
Let $\alpha_k$ denote the its expansion after $k$ steps. Then, $\alpha_2 =1$,
and\; $\alpha_{k+1} \leq 2\alpha_{k} + 1$. Indeed,\, let $x$ be the new point, and
assume it was connected to $y$. Then, for any previously exposed vertex $a$,
\[
\widetilde{d}(a,x) ~=~  \widetilde{d}(a,y) + d(x,y) ~\leq~
\alpha_{k}\cdot d(a,y) + d(x,y) ~\leq~\alpha_{k}\cdot\left(d(a,x) + d(x,y)\right) +  d(x,y) 
~\leq~ (2\alpha_{k} + 1)\cdot {d}(a,x),
\]
where the last inequality follows from the choice of $y$.\\
The recursive formula $\alpha_2=1$,\; $\alpha_{k+1} \leq 2\alpha_{k} + 1$\; implies \;$\alpha_n \leq 2^{n-1}-1$.
\end{proof}
This turns out to be rather tight:
\begin{theorem}\label{thm:tree-lb}
There exists a family $\cal F$ of metrics on $n$ points for which any online embedding
into a tree incurs a distortion of at least $2^{\lfloor (n-4)/2 \rfloor}$. 

All the metrics in $\cal F$
will be submetrics of $(C,d_C)$, the continuous cycle $C$ of length one, equipped with its shortest path metric. Moreover, they are (offline) embeddable into the line with distortion 3.
\end{theorem}
\begin{proof}
The adaptive adversary will choose the points of $C$ to be
exposed. The goal of the adversary will be to ensure that for every $k
\geq 2$, the tree metric induced by the first $2k$ exposed points,
 distorts the metric induced by $d_C$  by at least $2^{k-2}$. 
 
We start with the following simple claims. For a tree $T$, and $x,y \in V[T]$, 
let $P_T(x,y)$ denote the unique path between $x$ and $y$ in 
$T$.  It will be convenient to view $T$ and $P_T$ as continuous objects rather than
abstract graphs.
\begin{claim}\label{fact:tree}
Let $u_1,u_2,u_3,u_4$ be four vertices in a tree $T$. Then, either~ $P(u_1,u_2) \cap P(u_3,u_4) \neq \emptyset$,~ or  $P(u_2,u_3) \cap P(u_1,u_4) \neq \emptyset$.  \qed
\end{claim}

\begin{claim}\label{cl:tree-lb}
Let $p,q,r,s$ be points in a metric space such that $d(p,q), d(r,s) \leq \alpha$ while $d(p,r),$ $d(p,s),$ $d(q,r),$ $d(q,s) \geq \beta$. Assume that $p,q,r,s$ are embedded into a weighted tree $T$ by a mapping 
$\phi$, and it holds that  ~$P_T(\phi(p),\phi(q)) \cap P_T(\phi(r),\phi(s)) \,\neq\, \emptyset$.~ Then, 
$d_T$, the shortest path metric of $T$, distorts $d_C$ by at least $\beta/\alpha$. 
\end{claim}
\begin{proof}
It is readily verified by counting the edges in the suitable paths in $T$ that  
\begin{equation}
\label{eq:88}
d_T(\phi(r),\phi(p)) \;+\; d_T(\phi(p),\phi(s)) \;+\; d_T(\phi(s),\phi(q)) \;+\; d_T(\phi(q),\phi(r))
~~\leq~~
2d_T(\phi(r),\phi(s)) \;+\;  2d_T(\phi(p),\phi(q))\,.
\end{equation} 
The inequality can be strict, as the edges in the intersection
of $P_T(\phi(p),\phi(q))$ and $P_T(\phi(r),\phi(s))$ contribute four rather two times to the right hand side.
On the other hand,
\begin{equation}
\label{eq:99}
d_C(r,p) \;+\; d_C(p,s) \;+\; d_C(s,q) \;+\; d_C(q, r)) ~~\geq~~ {{4\beta} \over {4\alpha}} \cdot \left( 2d_C(r,s) \;+\;  2d_C(p,q) \right)\,.
\end{equation}
Comparing (\ref{eq:88}) and (\ref{eq:99}), and keeping in mind that the left hand side of (\ref{eq:88}) is 
at most the left hand side of (\ref{eq:99}) times the expansion of $\phi$, 
while~ $2d_T(\phi(r),\phi(s)) \;+\;  2d_T(\phi(p),\phi(q))$~ is at least ~$2d_C(r,s) \;+\;  2d_C(p,q)$~ 
divided by the contraction of $\phi$, one concludes that the product of the expansion and the contraction of $\phi$ is at least $\beta/\alpha$.
\end{proof}
The adversary will work in phases, exposing 4 points in the first phase,
and exposing $2$ points in each subsequent phase. Let $T_t$ denote the weighted tree constructed by 
the embedding algorithm after $t$ phases, and let $\phi$ be the corresponding embedding.

Let us parametrize $C$ by choosing an arbitrary reference point $p_0 \in C$, and defining $p_x \in C$ 
to be the point obtained by moving a distance $x$, $0\leq x < 1$, from $p_0$ along the cycle in the clockwise direction. Given a pair of (non-diametral) points $p_x,p_y \in C$, let $I_C[p_x,p_y]$ denote
the shortest length interval containing $p_x, p_y$. 

At the first phase, the adversary exposes the points $p_0, p_{0.25}, p_{0.5}, p_{0.75}$.
By Claim~\ref{fact:tree}, in the corresponding $T_1$, either the $\phi$-images of $I_C[p_0, p_{0.25}]$ 
and $I_C[p_{0.5}, p_{0.75}]$, or the $\phi$-images of $I_C[p_{0.25}, p_{0.5}]$ and 
$I_C[p_{0.75}, p_{0}]$, must intersect. Assume without loss of generality that the former occurs.

The strategy of the adversary is to establish at the end of (each) phase $t$ four points $x_t,x_{t+1},x_{t+2},x_{t+3}$ such that: \\
*~~~ $I_C[x_t, x_{t+1}] \subseteq I_C[p_{0}, p_{0.25}]$,~ $I_C[x_{t+2}, x_{t+3}] \subseteq I_C[p_{0.5}, p_{0.75}]$; \\
*~~~ $|I_C[x_t, x_{t+1}]| = |I_C[x_{t+2}, x_{t+3}]| = 2^{-t-1}$;\\
*~~~ the   $\phi$-images of $I_C[x_t, x_{t+1}]$ and $I_C[x_{t+2}, x_{t+3}]$ intersect in $T_t$.
\\ $\mbox{}$

If the adversary manages to achieve this, then, observing that $d_C(x_i,x_j) \geq 0.25$
for $i\in \{t,t+1\},\,j\in \{t+2,t+3\}$, and using Claim~\ref{cl:tree-lb}, one concludes that the
distortion of $\phi$ after $t$ phases is at least $0.25\cdot 2^{t+1} = 2^{t-1}$. Since the number of
points exposed after $t$ phases is $(2t+2)$, this implies the Theorem.

It remains to specify the strategy of the adversary. For $t=1$, the points $x_1,x_2,x_3,x_4$ are defined to be $p_0,\, p_{0.25},\, p_{0.5},\, p_{0.75}$, respectively. Assume inductively that by the end of phase 
$(t-1)$, there exists the quadruple $x_{t-1},x_{t},x_{t+1},x_{t+2}$ as
required. The new points
to be exposed by the adversary are  the 
median $a$ of $I_C[x_{t-1}, x_{t}],$ and the median $b$ of $I_C[x_{t+1}, x_{t+2}]$, respectively. 

Since the  $\phi$-images of $I_C[x_{t-1}, x_{t}]$ and $I_C[x_{t+1}, x_{t+2}]$ intersect in $T_{t-1}$,
they will intersect in $T_t$ as well. Then, the $\phi$-image of one of the intervals 
$I_C[x_{t-1}, a]$,~ $I_C[a, x_{t}]$, must intersect the $\phi$-image of one of the intervals 
$I_C[x_{t+1}, b]$,~ $I_C[b, x_{t+2}]$. Assuming, e.g., that~ 
$\phi (I_C[x_{t-1}, a]) \,\cap\, \phi (I_C[b, x_{t+2}]) \neq \emptyset$,~ set $x'_t=x_{t-1},\,x'_{t+1}=a,\,
x'_{t+2}=b,\,x'_{t+3}=x_{t+2}$.

We conclude the proof by observing that "cutting" $C$ at the point $p_{0.51}$ to obtain an interval, one obtains an embedding of any metric in $\cal F$ into this interval, and hence into the line, with distortion bounded by $3$.
\end{proof}
%
%
%
\section{Online embedding into the line}\label{sec:line}
In this section we show that there exists an online algorithm for embedding arbitrary 
$n$-point metric spaces into the line with exponential distortion. Since the line is a very special case of 
a tree, and even the tree metrics do not embed well into it, this can be viewed as a significant strengthening of Theorem~\ref{cl:tree-emb}. The upper bound is, however, weaker.
\begin{theorem}\label{thm:line}
Any metric space $(X,d)$ of size $n$ can be online embedded  into the line
with distortion bounded by $O(n 4^n)$. A priori knowledge of $n$ is not required.
\end{theorem}
\begin{proof}
The online embedding $\phi$ of $(X,d)$ into the line is constructed in the following inductive manner.
Let $x_k \in X$ be the point exposed at stage $k$.
\begin{definition}
\label{def:line}
Set $\phi(x_1) = 0$. For $k>1$, let $z$ be the closet neighbor of $x_k$ among the previously exposed points. Let $I^{x_k}$ be the leftmost open interval of length $2^{-k}d(z,x_k)$ to the right of $z$ that contains no previously exposed points. Set $\phi(x_k)$ to be the middle of this interval.
\end{definition}
In what follows, we call such $z$ the {\it father} of $x_i$, and denote it by $z=\father(x_k)$.  
Observe that since there are only $(k-1)$ points exposed prior to
$x_k$, there must exist such an empty interval starting at distance at most $(k-2)\cdot 2^{-k}d(x_k,z)$ away from $\phi(z)$.
Therefore,
\begin{equation}
\label{eq:49}
0 ~~\leq~~  \phi(x_k) - \phi(\father(x_k)) ~~\leq~~ (k-1.5)\cdot 2^{-k}d(x_k,\father(x_k))\,.
\end{equation}

{\em \underline{Bounding the expansion:}}\\
Let $\gamma_k$ denote the expansion of $\phi$ after the embedding of $x_k$.
The upper bound on $\gamma_k$ is obtained by an inductive argument similar to the one
used in Theorem~\ref{cl:tree-emb}.
Let $z=\father(x_k)$, and let $y$ be any previously exposed point. By definition of $z$,~~ 
$d(x_k,z) \leq d(x_k,y)$, and  ~$d(y,z) \leq d(x_k,y) + d(x_k,z) \leq 2d(x_k,y)$.
Combining this with (\ref{eq:49}), one gets
\[
|\phi(x_k)-\phi(y)| ~\leq~ |\phi(x_k) -\phi(z)|  \;+\; |\phi(z) -\phi(y)| ~~\leq
~~  (k-1.5) \cdot 2^{-k} d(x_k,z) \;+\; \gamma_{k-1}d(y,z)  ~~\leq
\]
\[
\leq ~~ k 2^{-k}\cdot d(x_k,y)  + 2\gamma_{k-1}d(x_k,y) ~~=~~ (k 2^{-k} +   2\gamma_{k-1})\cdot d(x_k,y)\,.  
\]
Thus, one gets a recursive equation~~ $\gamma_k \leq k2^{-k} + 2 \gamma_{k-1}$.~
Together with $\gamma_2 = 1$, this implies~  $\gamma_k \leq  2^{k+1}$.\\ $\mbox{}$

%
%
{\em \underline{Bounding the contraction:}}\\
We start with the following general claim:
\begin{claim}
\label{cl:collective}
Let $I^x$ be the open interval with $\phi(x)$ in the middle, and let
$I_L^x$ and $I_R^x$ be, respectively, its left and right (open)
halves. Assume that at some time in the process of exposure, and after
$x$ is exposed, $I_L^x$ had $p$ hits, and  $I_R^x$ had $r$ hits. (By a hit we mean an assignment
of $\phi(y)$, $y\in X$, to a point in the interval.) Then,\vspace{0.25cm}
 
* $I_L^x$ contains an empty open interval of length \;$\geq |I_L^x|/2^p$ that is adjacent to $\phi(x)$.
\vspace{0.25cm}

* $I_R^x$ contains an empty open interval of length \;$\geq |I_R^x|/(r+1)~\geq~
|I_R^x|/2^r$.
\end{claim}
\begin{proof}
The first statement is easily established using the induction on $p$. 
The second statement follows at once from the observation that any $r$ points in $I_R^x$
split it into at most $(r+1)$ parts. 
\end{proof}
%
Consider any pair $a,b$ of exposed points. By definition of the 
embedding algorithm, there are two sequences with the same initial element $z=y_1, y_2, \ldots ,
y_q=a$ and $z=w_1, w_2, \ldots, w_f =b$, where $y_i=\father(y_{i+1}), ~w_j
= \father(w_{j+1})$ and $\{y_i\}_{2}^q,\; \{w_j\}_2^f$ are disjoint.
In what follows, for $x\in X$ it will be convenient to use 
$\ord(x)$ to denote the time of $x$'s exposure. 

Let $\delta_i = d(y_{i-1},y_{i}), i=2, \ldots, q$\; and\; $\zeta_j =
d(w_{j-1},w_{j}), j=2, \ldots, f$. Let $D$ be the maximum among all $\delta_i$'s and $\zeta_j$'s,
and assume without loss of generality that  $D=d(y_{k-1},y_k)$. 
Thus,  $\phi(y_k)$ is the middle of an empty interval 
$I=I^{y_k}$ of length  $2^{-\ord(y_k)}D$ starting at $\phi(y_{k-1})$. Let $I_L = I_L^{y_k}$ and $I_R = I_R^{y_k}$ the left and the right halves of $I$.

\begin{claim}\label{cl:53}
For $i > k$,\;  it holds that~ $\phi(y_{i}) \in I_R$, ~and moreover,~~ 
$\phi(y_{i}) - \phi(y_k) \;\leq\; \frac{|I|}{2} \cdot (1-2^{\ord(y_k) - \ord(y_i)})$.
\end{claim}
\begin{proof}
We claim that for every $y_i$, $i\geq k$, at the end of the stage $\ord(y_i)$,
there exists an empty interval $J^i \subseteq I_L$ of length~ 
$\geq {{|I|} \over {2}} \cdot 2^{\ord(y_k) - \ord(y_i)}$ that lies entirely to the right of 
$\phi(y_i)$.
Indeed, this is the case for $i=k$, with $J^k=I_R$. Assume inductively that the claim 
holds for $(i-1)$. Then, prior to the arrival of $y_i$, the interval $J^{i-1}$ had at most 
$\ord(y_i) - \ord(y_{i-1}) - 1$ hits, and therefore, by Claim~\ref{cl:collective}, it contains
a subinterval of length $|J^{i-1}|/2^{\ord(y_i) - \ord(y_{i-1}) - 1}$. By the inductive
hypothesis and the definition of $I$ and $I^{y_k}$,
\[
{{|J^{i-1}|} \over {2^{\ord(y_i) - \ord(y_{i-1}) -1}}} ~~\geq~~  {{|I|} \over {2}} \cdot {{2^{\ord(y_k) - \ord(y_{i-1})}} \over {2^{\ord(y_i) - \ord(y_{i-1}) - 1}}}
~~=~~
|I|  \cdot 2^{\ord(y_k) - \ord(y_i)} ~~=~~ D\cdot 2^{-\ord(y_i)} ~~\geq~~ |I^{y_i}|\,.
\]
This means that by the time of arrival of $y_i$, there is an empty interval of length 
$~\geq~ |I|  \cdot 2^{\ord(y_k) - \ord(y_i)} ~~\geq~~ |I^{y_i}|$ lying entirely
in $J^{i-1}$, and hence in $I$, and in particular to right of $y_{i-1}$, the father of $y_i$. Therefore, $y_i$ will be mapped to the left (with a possible equality) of the middle of  $J^{i-1}$, and the statement follows with $J^{i} = J_R^{i-1}$. 

To bound $\phi(y_{i}) - \phi(y_k)$, observe that 
\[
\phi(y_{i}) - \phi(y_k) ~\leq~ |I_R| - |J^i|  ~\leq~  \frac{|I|}{2} \cdot (1-2^{\ord(y_k) - \ord(y_i)})\,.
\]
\end{proof}
A very similar argument implies also the following claim. Its proof is very close to that of Claim~\ref{cl:53}, and is omitted.
\begin{claim}\label{cl:37}
Let $\{w_j\}_1^f \subset X$ be the sequence as above, with $w_f =b$, and assume that 
for some $s \leq f$,~ $w_s\in I_L$. Then, $b\in I_L$, and moreover, 
~$\phi(y_k) - \phi(b) \geq  \frac{|I|}{2} \cdot 2^{\ord(y_k) - \ord(b)}$\,. \noproof
\end{claim}
We are in a position to upper bound the contraction.
Consider the sequences $\{y_i\}_{1}^q,\; \{w_j\}_1^f$ as before, with $y_q=a$, $w_f = b$
Let $w^*$ be the last of $w_i$'s with $\ord(w^*) \leq \ord(y_k)$. Such an element exists since $z$ is the common ancestor of both $y_i$'s and and $w_j$'s. Recall that at the time of exposure of $y_k$, its interval $I$ was completely empty. There two possibilities to consider:

If $w^*$ was mapped to the left of $\phi(y_k)$, then, by Claim~\ref{cl:37}, so is $b$,
and it holds that $\phi(a) - \phi(b) \geq \phi(y_k) - \phi(b) \geq |I|/2\cdot 2^{\ord(y_k) - \ord(b)}$.

It $w^*$ was mapped to the right of $I$, then so is $b$, and by Claim~\ref{cl:53},
$\phi(b) - \phi(a) \geq   |I|/2 \cdot 2^{\ord(y_k) - \ord(a)}$.

In both cases,~ $\phi(b) - \phi(a) \geq   D/2 \cdot 2^{-n}$. However,
\[
d(a,b) ~\leq~ d(a,z) + d(b,z) ~\leq~ \sum_{i=2}^q  d(y_{i-1},y_i) +  \sum_{j=2}^f  d(w_{j-1},w_j) ~\leq~ n D\,.
\]
Combining the two observation, one concludes that the contraction is at most 
$n2^{n+1}$, and hence the distortion is $O(2^n \cdot n 2^n) = O(n4^n)$.
\end{proof}
%
%
\section{Online embedding into $\ell_{\infty}$ with $(1+\epsilon)$ distortion}
In this section we study online embeddings of metric spaces $(X,d)$ into
$\ell_\infty$  with distortion bounded by $(1+\epsilon)$, where $\epsilon$ is arbitrarily
small.  The goal is to keep the dimension of the host $\ell_\infty$ space as small as possible. We provide some upper and lower bounds on the dimension in terms of
$|X|$ and $\epsilon$.

It is well known that {\em any} metric space $(X,d)$ with $|X|=n$
points can be isometrically embedded into
$\ell_{\infty}^{n}$. The details of the canonical embedding achieving
this will be discussed in a moment. Somewhat surprisingly, Th.~\ref{lem:old} below implies
that even the metrics on $4$ points cannot be online isometrically embedded into 
$\ell_\infty$ of {\em any finite   dimension}.  Moreover, achieving $(1+\epsilon)$-distortion
poses more difficulties that one could expect.

The main results of this section are:

\begin{theorem}
\label{th:finally}
For any $\epsilon \in (0,1]$, and (unknown in advance) input metric space $(X,d)$
with $|X|=n$, $(X,d)$ can be online  embedded in $\R^r$ equipped with the
$\ell_\infty$ metric with multiplicative distortion bounded by $(1+\epsilon)$, where  
$r=2^{O(n\log n)}\cdot (1+1/\epsilon)^n$.
\end{theorem}
\begin{theorem}
\label{lem:old} 
Consider any uniform online embedding procedure for 1-Lip embedding of
metric spaces $(X,d)$ with $|X|=4$, into $\ell_\infty^k$. There exists $(X_0,d_0)$, on which the resulting embedding $\Phi$ incurs a contraction of \; $(1+ {1 \over {2k+1}})$.  Consequently, to ensure a distortion  of at most 
$(1+\epsilon)$, the dimension must be $\Omega(\epsilon^{-1})$.
\end{theorem}
In what follows it will be convenient to identify an embedding $\Phi$ of a metric space $(X,d)$ into
 $\ell_{\infty}^{r}$ with a family\, $\Phi = \{\varphi_i\}_1^r$\, of embeddings of $(X,d)$ into
the line, where  $\varphi_i$ is the $i$-th coordinate of $\Phi$, and
$\|\Phi(x) - \Phi(y)\|_\infty = \max_i |\varphi_i(x) - \varphi_i(y)|$. 
Observe that $\Phi$ is $1$-Lip iff each $\varphi_i$ is. The canonical isometric embedding $\Phi$ of 
$(X,d)$ into $\ell_\infty^{n}$, $n=|X|$, can be described in these terms as~
$\Phi = \{\varphi_x\}_{x\in X}$,
where for each $x \in X$, $\varphi_x:X \rightarrow \R$ is defined by $\varphi_x(z) = d(x,z)$.
\subsection{Proof of Theorem \ref{th:finally}: The online embedding procedure and its analysis}
Let $\langle x_1,x_2,\ldots,x_n \rangle$ be the {\em exposure order} on $X$, specifying
the order of appearances of the points from $X$. 

Define $\Phi^t = \{\varphi^t_i\}$ as the family of $1$-Lip embeddings into the line of 
$X^t = \langle x_1,x_2,\ldots,x_t\rangle$. It will be convenient to describe the embedding procedure 
in a tree-like fashion, where the embeddings of $\Phi^t$ correspond to the vertices of the $(t-1)$'th level 
of the following tree. The root, $\Phi^1$, is a single embedding $\varphi^1$ sending $x_1$ to $0$. 
Once $x_t$ is exposed, each embedding $\varphi^{t-1} \in \Phi^{t-1}$,  splits (i.e., is extended) into a number of embeddings $\varphi^{t}\in \Phi^{t}$ according to the value assigned to $\varphi^{t}(x_{t})$. 
The values previously assigned to $X^{t-1}$ remain unchanged, i.e., for every child $\varphi^{t}$ of $\varphi^{t-1}$, it holds that $\varphi^{t}|_{X^{t-1}} = \varphi^{t-1}|_{X^{t-1}}$. Define $\Phi$ as $\Phi^n$. In the actual implementation, since splitting is impossible to implement online, all $\Phi^t$'s will be of the same size $|\Phi^n|$, containing suitably many copies of each embedding $\varphi^t$.

It is well known that a 1-Lip embedding $\phi^{t-1}$ of $X^{t-1}$ into the line can always be to extended 
to a 1-Lip embedding $\phi^{t}$ of $X^t$ into the line. More concretely, 
define the interval $I(x_t, \phi^{t-1}) = [\max_{y\in X^{t-1}} \phi^{t-1}(y) - d(y,x_t), \min_{y\in X^{t-1}} \phi^{t-1}(y) + d(y,x_t)]$. This interval is non-empty, since for any $y_1,y_2 \in X^{t-1}$ it holds that 
\[
[\phi^{t-1}(y_2) + d(y_2,x_t)]\, - \,[\phi^{t-1}(y_1) - d(y_1,x_t)] ~=~  [d(y_1,x_t) + d(y_2,x_t)] \,+\, [\phi^{t-1}(y_2) - \phi^{t-1}(y_1)] ~\geq
\]
\[
 d(y_2,y_1) \,-\, |\phi^{t-1}(y_2) - \phi^{t-1}(y_1)| ~\geq~ 0\,.
\]
It it readily verified that an assignment $\phi^t(x_t)=y$ results in a 1-Lip extension of $\phi^{t-1}$ iff
$y \in I(x_t, \phi^{t-1})$. 

Now, in order to describe the online embedding procedure, it suffices the specify the exact manner in which 
the children $\{\varphi^{t}\}$ of each $\varphi^{t-1} \in \Phi^{t-1}$ are produced. 
\subsubsection{The extension of $\varphi^{t-1}$ by $x_t$}
\label{def:constr} 
Let $\delta \in (0,1]$ be a fixed parameter, and let $X^{t-1}$ and $\varphi^{t-1}$ be as above. 
\begin{definition}
For all $i\leq t$ define the {\em scale} ~$s(x_i,x_t) = \delta \cdot d(x_i,x_t)$, and the corresponding infinite set~ $S(x_i, x_t) =  s(x_i,x_t) \cdot \Z$.
\\$\mbox{}$~~~~
Let $A^t   ~=~ \cup_{i<t}\, S(x_i,x_t)  ~\cup~ \left\{\varphi^{t-1}(x_i)\right\}_{i=1}^{t-1} $.~
$A^t$ will be called the {\em admissible range} of $\varphi^t(x_t)$, meaning that $\varphi^t(x_t)$
may only take values from $A^t$.
\\$\mbox{}$~~~~
Finally, define $P^t \subset A^t$, the {\em feasible range} of $\varphi^t(x_t)$, by ~$P^t  ~=~ A^t ~\cap~  I(x_t, \varphi^{t-1})$, as the set of the all the actual values that $\varphi^t(x_t)$ will take.
For $t=1$, let $P^t = \{0\}$.
\end{definition}
\begin{definition}
\label{def:F}
For each point $p \in P^t$, create an extension $\varphi^{t}$ of $\varphi^{t-1}$ by assigning 
$\varphi^{t}(x_t)=p$. If  $P^t$ is empty (a situation not to be encountered in the subsequent analysis), create an extension $\varphi^{t}$ of $\varphi^{t-1}$ by assigning $\varphi^{t}(x_t)=p$~
 for an arbitrary point \,$p \in I(x_t, \varphi^{t-1})$.  
\end{definition}
By definition of $I(x_t, \varphi^{t-1})$, each $\varphi^{t}$ is indeed a 1-Lip embedding of $X^t$ into the line.

Observe also that\,  $I(x_t, \varphi^{t-1})$ is contained in\, $[\varphi^{t-1}(x_i) -d(x_i,x_t),\;
\varphi^{t-1}(x_i) + d(x_i,x_t)]$,\, and so its length is at most $2 d(x_i,x_t)$. Consequently,\, 
$|S(x_i,x_t) \cap I(x_t, \varphi^{t-1})| \leq 2\delta^{-1}+1$,
implying that the number of different extensions $\varphi^{t}$ of
$\varphi^{t-1}$ is at most~  $(t-1)(2\delta^{-1} + 2)$.  
Thus,\; $|\Phi^n| \leq (n-1)! \cdot (2\delta^{-1} +2)^{n-1}$.

This completely defines the online embedding procedure. Observe that $x_1$ will be mapped 
to $0$ by any $\varphi \in \Phi$.
\subsubsection{The weak additive approximation property} 
%
In this section we state and prove the key property of the above embedding $\Phi$. 

Let $\psi: X \rightarrow \R$ be an arbitrary 1-Lip embedding of $(X,d)$ into the line. 
Let $\langle  z_{1}, z_{2}, \ldots, z_n \rangle$ be a reordering of
$X$ such that $\psi(z_{i}) \leq \psi(z_{i+1}), ~ i=1, \ldots ,n-1$. We
refer to this order as the  {\em $z$-order} on $X$. This to distinguish it from the exposure
order on $X$, to be also called the $x$-order.

For $i < j$, let $\D{ij} = \max_{x,y \in \langle z_i, \cdots,z_j\rangle} d(x,y)$.
\begin{lemma} 
\label{lm:main}
For any 1-Lip $\psi: X \rightarrow \R$, and any pair of indices $i < j$, there exists 
$\widetilde{\psi} \in \Phi$  such that,
$$\wp(z_j) - \wp(z_i) ~~\geq~~ \psi(z_j) - \psi (z_i) - 10n^2 \delta \cdot \D{ij}\,,
$$ 
where $\delta$ is the parameter used in the construction of $\Phi$.
\end{lemma}
\begin{proof} 
We proceed by considering increasingly more complex special cases. {\em Case 1} is in a sense fundamental, whereas other cases will reduce to it via a manipulation of the given 
$\psi$, and its replacement by a suitable 1-Lip $\psi^*$ that preserves the values of $\psi$ on
$z_i$ and $z_j$. 

The points $z_i, z_j$ and their values $L = \psi(z_i)$ and $R=\psi(z_j)$ are fixed, and will
not be effected by the changes in $\psi$, and, respectively, in the $z$-order (although the old indices $i,j$ may become irrelevant). 

Let $S^\psi$ be the set containing the members of the $z$-order interval 
$\langle z_i, z_{i+1},\ldots,z_{j} \rangle$. Observe that the interval
$[L,R]$ may contain images of points $x\not\in S^\psi$, e.g., some
$z_k$ with $k<i$, but $\psi(z_k)=\psi(z_i)$. Such points will be
called {\em incidental} points, and typically will require a special attention.

Instead of $\D{ij}$, the $d$-diameter of $S^\psi$, it will be sometimes more convenient to use the notation $D^\psi$. Let $z_f$, or  $z_f^\psi$, denote the first  point {\em in the exposure order} among $S^\psi$. A change in the $z$-order may cause a change in $z^\psi_f$, as well as  in $S^\psi$ and $D^\psi$; the latter two may only grow.
\vspace{0.3cm}\\ $\mbox{}$
\underline{\bf * Case I:}~ $z_i = z_f =x_1,\;z_j=z_n$. \vspace{0.1cm}\\ Assume w.l.o.g., that $\psi(x_1)=0$, as satisfied by all $\varphi \in \Phi$. Otherwise,
consider $\psi - \psi(x_1)$ instead of $\psi$. \\
We start with setting {\em targets} $t(z_k) \in \R$ for every $z_k$,
being the candidates for the value of $\wp(z_k)$. The logic behind this definition
will become apparent in Claim~\ref{obs:11}.
\begin{definition}
\label{def:target}
For $1 \leq k\leq f$, set \,$t(z_k) = 0$. ~ For $k > f$,
set \,$t(z_k)$\, to be the largest value in $S(z_f,z_k)$ such that
\[
t(z_k) ~\leq~ \psi(z_k) -  \delta\cdot (k-f-1)\,\D{ij}\,.
\]
In particular, for $k \geq f$, the further is $\psi(z_k)$ from $\psi(z_f)$, the more
will $t(z_k)$ shift away from $\psi(z_k)$ to the left towards $0$. 
\end{definition}
Observe that the targets $t(z_k)$ are admissible values for $z_k$, since
the value $0$ is always admissible, and for every $k \geq f$,~ $t(z_k) \in S(z_f,z_k)$. 
However, in general, $t(z_k)$ are not feasible.
\ignore{
}
To address this problem, $\wp$ will be defined in an inductive (offline) manner according to the 
exposure order (rather than the $z$-order) of $X$, taking care to preserve the 
$z$-order, i.e.,~ if $z_i$ precedes $z_j$ in the $z$-order, then \;$\wp(z_i) \leq \wp(z_j)$.

For $k >  f$, let $z_k^-$ and $z_k^+$ be, respectively, the maximal point in $X$  that precedes $z_k$, and the minimal point in $X$ that succeeds $z_k$, both according to the exposure order.
While $z_k^+$ may not be defined, $z_k^-$  always is, as $k < f$. The formal definition of 
$\wp$ is as follows:
\begin{definition}
\label{def:appr}
For $k \leq f$, set $\wp(z_k)=0$. \\
For $k \geq f$, by the inductive assumption about the preservation of the $z$-order,~
it holds that~ $\wp(z_k^-) \leq \wp(z_k^+)$. The value of $\wp(z_k)$ is defined according 
to how the value of  $t(z_k)$ relates to those of $\wp(z_k^-)$ and $\wp(z_k^+)$:

\noindent
{\bf if}~ $\wp(z_k^-) \;\leq\; t(z_k) \;\leq\; \wp(z_k^+)$ ~{\bf then}~
$\wp(z_k) = t(z_k)\;;$\\
{\bf if}~ $t(z_k) \;<\; \wp(z_k^-)$ ~{\bf then}~
$\wp(z_k) = \wp(z_k^-)\;;~ ~ ~ $  {\em (* "pull back" of \,$t(z_k)$\, away from $0$~*)}  \\
{\bf if}~ $t(z_k) \;>\; \wp(z_k^+)$ ~{\bf then}~
$\wp(z_k) = \wp(z_k^+)\,.~ ~ ~ ~$  {\em (* "push forward" of \,$t(z_k)$\, towards $0$~*)}  \\
\ignore{
  $\wp(x_1)~\leftarrow~ 0$;\\
$\mbox{}$\hspace{0.5cm}
   {\bf for} $t=2$\; {\bf to}\; $n$\; {\bf do}\\
   $\mbox{}$\hspace{1cm}
  $m^{t-1}_-  ~\leftarrow~ \max\{\wp(x_i)\;|\; x_i \in X_-^{t-1}\}$;\\
   $\mbox{}$\hspace{1cm}
   $m^{t-1}_+  ~\leftarrow~ \min\{\wp(x_i)\;|\; x_i \in X_+^{t-1}\}$;~~~~~~~~~~~~~~~ {\em (*~~ it will always hold that ~$m^{t-1}_- \,\leq\, m^{t-1}_+$~~ *)}\\
   $\mbox{}$\hspace{1cm}
   {\bf if}~ $\widetilde{x_t} < m^{t-1}_- $~~ {\bf set}~~ $\wp(x_t) ~\leftarrow~ m^{t-1}_- $;~~~~~~~~~~~~~~ {\em (*~~ "pull back" ~~*)}\\
   $\mbox{}$\hspace{1cm}
   {\bf if}~ $\widetilde{x_t} > m^{t-1}_+ $~~ {\bf set}~~ $\wp(x_t) ~\leftarrow~ m^{t-1}_+ $;~~~~~~~~~~~~~~ {\em (*~~ "push forward" ~~*)}\\
    $\mbox{}$\hspace{1cm}
   {\bf if}~ $\widetilde{x_t} \in  [m^{t-1}_-,\, m^{t-1}_+]$ ~~ {\bf set}~~ $\wp(x_t) ~\leftarrow~ \widetilde{x_t}$\.,~~~~~ {\em (*~~ "free jump" ~~*)}\\
    $\mbox{}$  \hspace{2cm}
}
\end{definition}
Clearly, it holds that~ $\wp(z_k^-) \leq \wp(z_k) \leq \wp(z_k^+)$,~ and thus the $z$-order is
indeed preserved.

Next, we claim that $\wp \in \Phi$. Observe that $x_1$ precedes or is equal to $z_f$ in the 
$z$-order, and thus $\wp(x_1)=0$, as required. For every $z_k \neq x_1$, the value $\wp(z_k)$ as defined above is admissibles for $z_k$, since $t(z_k)$ was such. It remains to show that it is also feasible, which in the present context is the same as showing that  $\wp$ is 1-Lip.  
\begin{claim}\label{obs:11}
For every $k >1, ~$ $~0~\leq~ \wp(z_k) - \wp(z_{k-1}) ~\leq~
\psi(z_k) - \psi(z_{k-1})$.   
\end{claim}
\begin{proof}
If $k \leq f$, the statement is trivial, since in this case $\wp(z_k) = \wp(z_{k-1}) = 0$.
Consider $k > f$. If either $\wp(z_k)$ was defined by a "pull back", or $\wp(z_{k-1})$ was defined by
a "push forward", then, due to the $z$-order preservation,  $\wp(z_k) = \wp(z_{k-1})$,
and we are done.  Otherwise, it must hold that 
$\wp(z_k) \leq t(z_k)$, while $\wp(z_{k-1}) \geq t(z_{k-1})$. 
In this case, introducing the notation $\alpha_t = \psi(z_t) - \delta\cdot(t-f-1) \D{ij} - t(z_t) \geq 0$~ (compare with Definition~\ref{def:target}), one gets:
\[
\wp(z_k) - \wp(z_{k-1}) ~\leq ~ t(z_k) -
  t(z_{k-1}) ~=~ 
\psi(z_k) - \delta\cdot(k-f-1) \D{ij} -
  \alpha_k] ~-~ [\psi(z_{k-1}) - \delta\cdot((k-1)-f-1) \D{ij} -
  \alpha_{k-1}]
\]
\[
=~~ \psi(z_k) - \psi(z_{k-1}) ~-~  (\delta\cdot \D{ij} + \alpha_k  - \alpha_{k-1})
~~~\leq~~~ \psi(z_k) - \psi(z_{k-1})\,,
\]
where the last inequality holds  since ~ $\alpha_k \leq s(z_f,z_k) = \delta \cdot d(z_f, z_k) \leq \delta \cdot \D{ij}$,~ $0 \leq \alpha_{k-1}$, and so
 $(\delta \cdot \D{ij} + \alpha_k  - \alpha_{k-1}) \geq 0$.
\end{proof}
Claim~\ref{obs:11} immediately implies that $\wp$ is 1-Lip with respect to the line metric induced on $X$ by $\psi$, and hence it is 1-Lip with respect to the original $(X,d)$ as well.  

Finally, we need to show that $\wp(z_j) - \wp(z_i)$ additively
approximates $\psi(z_j) - \psi(z_i)$, as required by the lemma. 
By Definition~\ref{def:appr}, since $z_i=z_f=x_1$, it holds that $\wp(z_i)=0$, and since
$j=n$,~ it holds that $\wp(z_j) \geq t(z_j)$. Keeping in mind that by our assumptions $\psi(z_i) = \psi(x_1)=0$, and borrowing the notation $\alpha_t$ from the proof of Claim~\ref{obs:11}, we get
\begin{equation}
\label{eq:weakI}
\wp(z_j) - \wp(z_i) ~\geq ~ t(z_j) -  0 ~=~ [\psi(z_j) - \delta\cdot(j-f-1) \D{ij} -
  \alpha_j] ~-~ \psi(z_{i}) ~\geq~ \psi(z_j) - \psi (z_i) -  \delta \cdot (j-i)\D{ij} 
\end{equation}
$\mbox{}$~
\underline{\bf * Case Ia:}~ $z_i = z_f$,~  $x_1$ precedes $z_i$,~ $z_j=z_n$. 
\vspace{0.1cm}\\ 
The construction is similar to that of {\em Case I}, the only difference being that 
one needs to consider $\psi(x) - \psi(z_f)$ instead of $\psi(x)$, and for all $k \leq f$,~ 
set $\wp(z_k)=0$. It easy to verify that $\wp$ is indeed feasible.
\ignore{
{\bf Important Remarks:} {\em In what follows, we shall need slightly stronger versions of {\em Case I}.

$\bullet$~  Call $z_k$ a {\em stray} point, if  $k>f$ and $\psi(z_k) = \psi(z_f)$. After slightly changing the construction of $\wp$, the additive approximation property as stated in Eq.~(\ref{eq:weakI}) remains true when $\D{ij}$, the $d$-diameter of the set $\{z_i,z_{i+1},\ldots,z_j\}$, is replaced by ${ D}^*_{ij}$, the $d$-diameter of this set minus the stray points. 

The change is in Def.~\ref{def:target}, the definition of the targets. The original $f$ used there should be replaced with $f' \geq f$, the maximal index such that 
$\psi(z_{f'}) = \psi(z_f) = 0$.~ Specifically, set \,$t(z_k) = 0$\, for all $k \leq f'$, while for $k>f'$,\, set \,$t(z_k)$\, to be the largest value in \,$S(z_f,z_k)$\, such that~
$t(z_k) ~\leq~ \psi(z_k) -  \delta\cdot (k-f'-1)\,D_{ij}^*$\,. The rest of the argument
remains essentially unchanged.

$\bullet$~  Consider any pair of points $z_q,z_r$ with $f \leq q < r \leq n$. Then,
\[
\wp(z_r) - \wp(z_q) ~\geq ~  \psi(z_r) - \psi (z_q) -  \delta \cdot (n-f)\cdot \D{ij}\,.
\]
To see this, observe that the minimal value that $\wp(z_r)$ can possibly get
is $\psi(z_r) - (n-f-1)\D{ij} - \alpha_r$, and that $\wp(z_q) \leq \psi(z_q)$.  
}
}
\vspace{0.4cm} 
\\$\mbox{}$~
\underline{\bf * Intermission: Useful operators on $\psi$}\\
All the below operators preserve the values of $x\in S^\psi$, and in particular, preserve the interval $[L,R] \subset \R$. 
%
\begin{definition}
\label{def:opera} {\em
$\mbox{}$
\begin{itemize}
\item {\em L-truncation} and {\em R-truncation}.\\
The L-truncation of $\psi$ is defined as $\psi^*$, where $\psi^*(z_k) = \psi(z_k)$ if
$k \geq i$,~ and $\psi^*(z_k) = \psi(z_i)$ if $k \leq i$.\\
The R-truncation of $\psi$ is defined as $\psi^*$, where $\psi^*(z_k) = \psi(z_k)$ if
$k \leq j$,~ and $\psi^*(z_k) = \psi(z_j)$ if $k \geq j$.

These operators preserve the original $z$-order, and the 1-Lip property. The set $S^\psi$ remains unchanged.
%
\item {\em L-expansion} and {\em R-expansion}.
\\
The L-expansion with respect to $S^\psi$ is defined as $\psi^*$, where 
$\psi^*(z_k) = \psi(z_k)$ \,if\, $z_k \in S^\psi$, or $\psi(z_k) > L$.
Otherwise, $\psi^*(z_k) = \max_{y\in S^\psi} \psi(y) - d(y,z_k)$. \\
The R-expansion with respect to $S^\psi$ is defined as $\psi^*$, where 
$\psi^*(z_k) = \psi(z_k)$ \,if\, $z_k \in S^\psi$, or $\psi(z_k) < R$. Otherwise,
$\psi^*(z_k) = \min_{y\in S^\psi} \psi(y) + d(y,z_k)$. \\
{\em Call $y_0 \in S^\psi$ that actually defines $\psi^*(z_k)$ the {\em anchor} of $z_k$.
In particular, the anchor of $z_k \in S^\psi$ is $z_k$ itself.} 

The set $S^\psi$ remains unchanged.
\item {\em L-folding} and {\em R-folding}.\\
The L-folding of $\psi$ is defined as $\psi^*$,~ where $\psi^*(x) = \psi(x)$ \,if\,
$\psi(x) \geq L$,\; and\; $\psi^*(x) = L + (L-\psi(x))$,\, i.e., the reflexion of $\psi(x)$ with respect to $L$, otherwise. \\
The R-folding of $\psi$ is defined as $\psi^*$,~ where $\psi^*(x) = \psi(x)$ \,if\,
$\psi(x) \leq R$,\; and\; $\psi^*(x) = R - (\psi(x) -R)$,\, i.e., the reflexion of $\psi(x)$ with respect to $R$, otherwise.

It is easy to verify that foldings preserve the 1-Lip property.
However, they may significantly change the $z$-order. For $\psi^*$ obtained from $\psi$ by a folding, the set $S^{\psi^*}$ is obtained from $S^\psi$ by augmenting it by all points $x$, such  that  $\psi^*(x) \in  [L,R]$ but $\psi(x) \neq \psi^*(x) $.
 
\item {\em Expanded  L-folding} and {\em R-folding}. \\
The expanded R-folding of $\psi$ is defined as $\psi^*$ resulting from the R-expansion with respect to $S^\psi$, followed by the R-folding. The expanded L-folding is defined analogously.

The set $S^{\psi^*}$ is obtained from $S^\psi$ as in the simple folding.

\item {\em Expanded truncated  L-folding} and {\em R-folding}. \\
The expanded R-folding of $\psi$ is defined as $\psi^*$ resulting from the R-expansion with respect to $S^\psi$, followed by the R-folding, and then collecting
all $z$'s with $\psi(z) \geq R$ whose current value is $< L$, and resetting this value to $L$.
The expanded L-folding is defined analogously.

The set $S^{\psi^*}$ is obtained from $S^\psi$ as in the expanded folding, however the
points whose value was reset in the last step are not included, making them {incidental}. 
As we shall see, expanded truncated foldings preserve the 1-Lip property.
\end{itemize}
}
\end{definition}
\begin{claim}
\label{ob:e}
Let $\psi$ be 1-Lip, and let $\psi^*$ be its R-expansion. Then,~ $\forall x$ $\psi^*(x) \geq \psi(x)$,~ and $\psi^*|_{S^\psi \,\cup\, \{x \;|\; \psi(x) \geq R\}}$ ~is 1-Lip with 
respect to $d$.
A similar statement holds for the L-expansion.
\end{claim}
\begin{proof} For the first statement, it suffices to verify it for $x$'s with $\psi(x) \geq R$. 
Keeping in mind that $\psi$ is 1-Lip, and using the anchor $x_0 \in S^\psi$ of $x$, one has
\[
\psi^*(x) - \psi(x) ~=~ [\psi(x_0) + d(x_0,x)] - \psi(x) ~=~ d(x,y) - [\psi(y) - \psi(x)] \geq 0\,.
\]
For the second statement, let~ $z,y \in \; S^\psi \,\cup\, \{x \;|\; \psi(x) \geq R\}$, and let $z_0$ be the anchor of $z$. Then,
\[
\psi^*(z) - \psi^*(y) ~\leq ~ [\psi(z_0) + d(z_0,z)] -  [\psi(z_0) + d(z_0,y)] ~=~ 
d(z_0,z) - d(z_0,y) ~\leq~ d(x,y)\,.
\]
\end{proof}
%
%
\begin{claim}
\label{ob:ef}
Let $\psi$ be 1-Lip, and let $\psi^*$ be its expanded R-folding. 
Then,~ $\psi^*|_{S^{\psi^*}}$ ~is 1-Lip with respect to $d$, and
$D^{\psi^*} \leq D^{\psi} + 4(R-L)$. A similar statement holds for the expanded L-folding.
\end{claim}
\begin{proof} The first statement follows at once from the previous claim, and the
1-Lip preservation property of the R-folding. 

For the second statement, observe that $x\in S^{\psi*}\setminus S^\psi$ if and only if ~$R< \psi^*(x) \leq  2R-L$. Let $x_0$ be the anchor of $x$. Then, as $\psi(x_0) \geq L$\; and\;
$\psi^*(x) =  \psi(x_0) + d(x_0,x) \leq 2R - L$, one concludes that ~$d(x_0,x) \leq  2(R-L)$.
Since $S^{\psi*}$ was obtained from $S^\psi$ by augmenting it with points at $d$-distance  at most $2(R-L)$ from $S^\psi$, the statement follows.
\end{proof}
Since truncation ensures that $\psi^(x)$ for $x\in S^{\psi*}\setminus S^\psi$ can only
get closer to $\psi(y)$'s with $y \leq L$ without ever jumping over them, we get
\begin{corollary}
\label{ob:tef}
Let $\psi$ be 1-Lip, and let $\psi^*$ be its  expanded truncated R-folding. 
Then,~ $\psi^*$ ~is 1-Lip with respect to $d$, and $D^{\psi^*} \leq D^{\psi} + 4(R-L)$. A similar statement holds for the expanded L-folding.
\end{corollary}
We are now ready to return to analysis of the remaining cases.
\vspace{0.5cm} $\mbox{}$\\ $\mbox{}$~ 
\underline{\bf * Case II:}~ {\em $z_i=z_f$, $\psi(x_1)=\psi(z_f)$,~ $z_j$ arbitrary.}\vspace{0.1cm}\\
The difficulty in using the construction of {\em Case Ia}\, is that \;$\D{in}$\; can now be considerably larger than $\D{ij}$. To resolve this difficulty, we replace the original 
$\psi$ by its expanded truncated R-folding, and then apply to it L-truncation. 
By Corollary~\ref{ob:tef} both operations preserve 1-Lip property, and so 
the resulting $\psi^*$ is 1-Lip, and $D^{\psi^*}$, the $d$-diameter of  
of $S^{\psi^*}$, is at most $5\D{ij}$. Hence, replacing $\psi$ by $\psi^*$ reduces the original problem to {\em Case Ia}, and implies the existence of $\widetilde{\psi^*} \in \Phi$,
such that 
\[
\widetilde{\psi^*}(z_j) - \widetilde{\psi^*}(z_i) ~~\geq~~ \psi(z_j) - \psi (z_i) -  \delta\cdot 5(j-i)\D{ij}\,.
\]
\vspace{0.3cm}\\ $\mbox{}$~
\underline{\bf * Case III:}~ $z_f=x_1 \in [z_i, z_j]$ in the $z$-order.  \\
We assume without loss of generality that $\psi(x_1)=0$.
This case reduces to\, {\em Case II}\, by applying the expanded truncated L-folding
to $\psi|_{\langle z_1,\ldots, x_1\rangle}$ with $z_i$ and $z'_j= x_1$, and 
the expanded truncated R-folding to $\psi|_{\langle x_1,\ldots, z_n\rangle}$
with $z'_i = x_1$ and $z_j$, respectively, and taking the union of the resulting partial functions. This is possible since their domains of definition share only the point $x_1$, and both
partial functions assign the same value $0$ to it.
The easy verification of that thus constructed $\psi^*$ is 1-Lip, and it does not change the
values of the original $z_i,z_j$, is left to the reader. The conclusion is
\[
\widetilde{\psi^*}(z_j) - \widetilde{\psi^*}(z_i) ~~\geq~~ 
[\psi(z_f) - \psi (z_i) -  \delta\cdot 5(f-i) \cdot \D{if}] ~+~
[\psi(z_j) - \psi (z_j) -  \delta\cdot 5(j-f) \cdot \D{fj}]  ~~\geq~~\]
\begin{equation}
\label{eq:CaseIII}
~~\geq~~ \psi(z_j) - \psi (z_i) -  \delta\cdot 5(j-i) \cdot \D{ij}\,.
\end{equation}
$\mbox{}$ \vspace{0.05cm}\\ $\mbox{}$
\underline{\bf * Case IV:}~ $x_1 \not\in [z_i,z_j]$ in the $z$-order.
\vspace{0.1cm}\\
In this case, we perform repeatedly expanded L-foldings followed by expanded R-foldings, creating a series of mappings $\psi_0=\psi,\psi_1,\psi_2,\ldots$, where the $i$'th folding is with respect to $S^{\psi_{i-1}}$. The process stops when $S^{\psi_{t}} = S^{\psi_{t+1}}$.

Observe that as long as the process runs, $|S^{\psi_i}|$ strictly grows, and therefore it must stop in at most $(n-2)$ rounds. Observe also that by Claim~\ref{ob:ef}, all the partial mapping 
$\psi^i|_{S^{\psi_{i}}}$  are 1-Lip. Finally, observe that Claim~\ref{ob:ef} implies
at the stopping time $t$, since $S^{\psi_{t}} = S^{\psi_{t+1}}$,  for every 
$x\not\in S^{\psi_t}$, the $d$-distance $\dist(x,S^{\psi_t}) > 2(R-L)$. 

\vspace{0.1cm}$\mbox{}$
Define a new mapping\, $\psi^*$, where $\psi^*(x) = \psi_t(x)$ \,if $x\in S^{\psi_t}$, and \,
$\psi^*(x) = \psi_t(z_f^{\psi_t})$, otherwise. By the last observation about $\psi^t$,\; 
$\psi^*$ is 1-Lip.

By Claim~\ref{ob:ef},   
\[
D^{\psi_t} ~\leq~ D^{\psi} + 2t\cdot(R-L) ~<~ (2n-3)\cdot D_{ij}\,.
\]
We can now reduce the problem to {\em Case III} with a small alteration. 
The problem is that mapping all $x\not\in S^{\psi_t}$ to $\psi_t(z_f^{\psi_t})$\;
causes $S^{\psi^*}$ strictly contain $S^{\psi_t}$, possibly causing in turn an uncontrollable growth of $D^{\psi_t}$. To overcome this problem, we assume without loss of generality that
$\psi_t(z_f^{\psi_t}) = 0$, and proceed as in {\em Case III} for all $x\in S^{\psi_t}$, while setting $\widetilde{\psi^*}(x) = 0$ for all $x\not\in S^{\psi_t}$. It is easy to see
that the resulting $\widetilde{\psi^*}$ is feasible. Thus, that there exists 
$\widetilde{\psi^*} \in \Phi$ such that 
\begin{equation}
\label{eq:CaseIV}
\widetilde{\psi^*}(z_j) - \widetilde{\psi^*}(z_i) ~~\geq~~ \psi(z_j) - \psi (z_i) -  \delta\cdot 5(j-i)\cdot 2n \cdot \D{ij}\,.
\end{equation}
This completes the proof of the Lemma.
\end{proof}
\subsubsection{The conclusion} 
We can now conclude the proof of Theorem~\ref{th:finally}.

As it was mentioned earlier, the embedding  $\varphi_x$ of $(X,d)$ into the line given by 
$\varphi_x(z) = d(x,z)$, 
is on one hand 1-Lip, while on the other hand, it preserves all distances between $x$ and the other points of $X$.
Thus, applying Lemma~\ref{lm:main} to $\varphi_x$ with $z_i=x$ and $z_j=y$, observing 
that $\D{ij} \leq 2d(z_i,z_j)$, and choosing $\delta = \epsilon/20n^2$,
we conclude that there exists $\widetilde{\varphi_x} \in \Phi$, such that 
\[
\widetilde{\varphi_x}(y) - \widetilde{\varphi_x}(x) ~~\geq~~  [\varphi_x(y) - \varphi_x(x)] - \delta \cdot 
10n^2 \cdot \D{xy} ~~=~~ d(x,y) - \epsilon \cdot \D{xy}/2 ~~\geq~~ (1-\epsilon)\cdot d(x,y).
\]
Ranging over all $\varphi_x$ and $y$'s, \;$x,y \in X$, we arrive at the desired conclusion.
By the discussion following the Definition~\ref{def:F}, the dimension of the host space 
is bounded from above by
\[
(n-1)! \cdot (2\delta^{-1} +2)^{n-1} ~\leq~ 40^n\,(n-1)!\, n^{2n}\,  (1+ 1/\epsilon)^n\,.
\]

\subsection{The Lower Bound}
In what follows we rove Theorem \ref{lem:old}.

 Consider the class $\cal A$ of metrics on $X=\{a,b,c,q\}$,
induced by the geodesic metric of $\S^1$, the continuous $1$-dimensional cycle, 
on two pairs of antipodal points, $\{a,b\}$ and $\{c,q\}$. The distances $d(a,b)$ and $d(c,q)$ are
 are equal to $1$. The distance $\alpha =d(a,c)=d(b,q)$ fully
 specifies a metric in  $\cal A$. Let us denote it by $d_\alpha$. 
\begin{theorem}
\label{lem:174} 
Let $\phi$ be any 1-Lipschitz online embedding of the class of metrics $\cal A$  as above
into $\ell_\infty^k$. Then, there exists $d_\alpha \in \cal A$ on which $\phi$ incurs a contraction of \;
$(1+ {1 \over {2k+1}})$.  Consequently, to ensure a distortion  of \; $\leq  1+\epsilon$\; for all $d_\alpha \in A$,   the dimension must be $\Omega(\epsilon^{-1})$.
\end{theorem}
\begin{proof}
The adversary will expose first the antipodals $a,b$, at distance 1. Then, seeing the embedding
of these two points, the adversary will choose $\alpha$ appropriately, and expose the antipodal $c,q$ 
according to $d_\alpha$.  

As before, we shall view the embedding $\Phi:X \rightarrow \ell_\infty^k$ as a family of $k$ online non-expanding embeddings of $X$ into the line, 
$\Phi=\{\phi_1, \ldots ,\phi_k, ~ ~ \phi_i : X \longrightarrow \R
\}$. We  assume, w.l.o.g.,  that $\phi_i(a) = 0$ 
for all $\phi_i \in \Phi$,  and that $\phi_i(a) \leq \phi_i(b) \leq 1$.
For convenience, assume also that $\phi_i$'s are ordered by the value of $\phi_i(b)$, i.e., that~ $0 \leq \phi_1(b) \leq \phi_2(b) \leq \ldots \leq \phi_k(b) \leq 1$.
The values $\{ \phi_i(b) \}$ partition $[0,1]$ into at most $(k+1)$ intervals, and so, 
by the pigeonhole principle, the middle-point of the longest of these intervals is at least ${1 \over {2(k+1)}}$ 
apart from any $\phi_i(b)$. Call this point $p$.
Let $j$ be such  $0 \leq \phi_{j}(b) < p < \phi_{j+1}(b) \leq 1$. By the choice of $p$,  it holds that
$p \geq \phi_i(b)+{1 \over {2(k+1)}}$ when $i \leq j$, and $p \leq \phi_i(b)-{1 \over {2(k+1)}}$
 when $i > j$.

The $\alpha$ chosen by the adversary will be $\alpha = {{1-p}\over 2}$.\; I.e., in the chosen 
$d_\alpha$, the distance between $a$ and $c$ is ${{1-p}\over 2}$.

We claim that for all $i$, whatever the values of $\phi_i(c)$ and $\phi_i(q)$ may be, ~
$|\phi_i(q) - \phi_i(c)| \leq 1 - {1 \over {2(k+1)}}$. \\ 
Indeed recall that $\phi_i$ is not expanding for all $i \in [k]$ and
hence, $\phi_i(c) \in [-\alpha, \alpha]$.

Consider first $i \leq j$ then $\phi_i(q) \in [\phi(b) \alpha, \phi +
\alpha$ and hence
so
\[
|\phi_i(q) - \phi_i(c)| ~\leq~ (\phi_i(b) + \alpha) - (\phi_i(a) - \alpha) ~=~ \phi_i(b) + 2\alpha
~\leq~  \phi_j + 1 - p  ~\leq~ 1 - {1 \over {2(k+1)}}~.
\]
Consider now $i \geq j+1$. Using the non-expansion of $\phi_i$ again, we conclude that
\[
   \phi_i(b) - (1 - \alpha)  ~\leq~ \phi_i(c)~\leq~ \phi_i(a) + \alpha\;; ~~~~~~~~~ 
\phi_i(b) - \alpha ~\leq~ \phi_i(q) ~\leq~ \phi_i(a) + (1 - \alpha)
\]
But since $\alpha \leq 0.5$ this implies that,
\[
|\phi_i(q) - \phi_i(c)| ~\leq~ \max\{ \, \phi_i(a) + (1 - \alpha) - (\phi_i(b) - (1 - \alpha)),~~
 \phi_i(a) + \alpha\ - (\phi_i(b) - \alpha)\, \} 
\]
\[
=~ \max\{\, 2(1-\alpha) - \phi_i(b),~~ 2\alpha - \phi_i(b) \,\} ~=~ 2 (1-\alpha) - \phi_i(b)
~=~ 1 + p - \phi_i(b) ~\leq~ 1 - {1 \over {2(k+1)}}~.
\]
Thus, the contraction of $\phi$ is at least $( 1 - {1 \over {2(k+1)}})^{-1} ~=~ 1 + {1 \over {2k+1}}$.
\end{proof}
%
\section{Remarks on isometric online embeddings of trees}
We conclude the discussion with a number of remarks on isometric embeddings of tree metrics. It is well known that the tree metrics $d_T$ (like the Euclidean metrics)  are 
rigid, i.e., there exists a unique (minimal) weighted tree $T^*$ with
Steiner points realizing $d_T$ as its submetric. Moreover, this $T^*$ can be constructed
in an online manner, regardless of the order of exposure.
\begin{theorem}\label{lem:l-1}
Every tree metric $d_T$ can be isometrically online embedded  into a metric of a weighted tree  $T^*$ (using Steiner points). The knowledge of $n$ is not required.
\end{theorem}
Skipping the details, at each step, given a new point $x$, the embedding introduces a new (uniquely defined) Steiner point $s_x$ in the tree constructed so far, and attaches $x$ to $s_x$ by a new edge of (uniquely defined) weight $w_x$. The Steiner point $s_x$
can also be a proper vertex of the tree.

Next we claim that tree metrics isometrically  embed online into
$\ell_1$ and $\ell_\infty$.
\begin{theorem}
  \label{lem:l-5}
  Every tree metric on $n$ points is isometrically online embeddable 
  into $\ell_1^{n-1}$.
\end{theorem}
The general idea is to follow the isometric embedding into a weighted tree as discussed above. The key property of the embedding is that the images of adjacent vertices of the tree will differ at a single coordinate. Using this, the image a new Steiner point $y_x$ lying between the existing adjacent vertices $v$ and $u$, will agree with both $v$ and $u$
on their coordinates of agreement, and differ from both on the remaining coordinate.
 
Attaching a new point $x$ to the corresponding Steiner point $y_x$
involves taking the vector representing $y_x$, and adding to it a new coordinate
with value $w_x$. The vectors constructed so far are assumed to have value $0$
on this coordinate. A relatively simple analysis show that this online embedding is
indeed an isometry.  
\\ $\mbox{}$

Our last remark is that anything online-embeddable into $\ell_1$ of an a priori known
dimension $D(n)$, online embeds into $\ell_\infty$ of dimension $2^{D(n)}$ as well.
It is so, since the mapping $x=(x_1,\ldots,x_D) \rightarrow (\langle x,\epsilon^1\rangle, \ldots, \langle x,\epsilon^{2^D}\rangle)$, where $\epsilon^i$'s range over all sign vectors of  \;$\pm 1$'s of dimension $D$, is an isometry between $\ell_1^D$ and $\ell_\infty^{2^D}$, and, moreover, it is online constructible. The dimension $2^{D}$ can be improved to 
$2^{D-1}$ by fixing the first coordinate of all $\epsilon_i$'s to be 1. Thus, 
\begin{theorem}
  \label{lem:l-4}
  Every tree metric on $n$ points is isometrically online embeddable 
  into $\ell_\infty^{2^{n-2}}$.
\end{theorem}

\bibliographystyle{plain}

\end{document}